%% file: prpox_arxiv.tex
\documentclass[11pt]{article}%
\usepackage{amsfonts}
\usepackage{fancyhdr}
\usepackage{comment}
\usepackage[a4paper, top=2cm, bottom=2.5cm, left=2cm, right=2cm]%
{geometry}
\usepackage{times}

\usepackage[ruled,linesnumbered,vlined]{algorithm2e}

\usepackage{times}
\usepackage{graphicx}
\usepackage{amsthm}
\usepackage{amsmath}
\usepackage{amsfonts,mathtools}
\usepackage{booktabs}
\usepackage{paralist}
\usepackage{subfigure}
\usepackage{color}

\newtheorem{theorem}{Theorem}[section]

\newtheorem{definition}{Definition}[section]
\newtheorem{lemma}{Lemma}[section]
\newtheorem{claim}{Claim}[section]

\usepackage{color}
\usepackage[capitalize]{cleveref}
\usepackage{amsthm}

\newenvironment{proofof}[1]{{\vspace*{5pt} \em Proof of #1:  }}{\hfill\rule{2mm}{2mm}\vspace*{5pt}}

\newcommand{\bR}{{\mathbb{R}}}
\newcommand{\MMS}{\mathsf{MMS}}
\newcommand{\PROP}{\mathsf{PROP}}
\newcommand{\WPROP}{\mathsf{WPROP}}
\newcommand{\APS}{\mathsf{APS}}
\newcommand{\bX}{\mathbf{X}}

\newcommand{\cR}{\mathcal{R}}
\newcommand{\cI}{\mathcal{I}}
\newcommand{\POF}{\textsf{PoF}\xspace}
\newcommand{\opt}{{\rm opt}}
\renewcommand{\sc}{{\rm sc}}

\begin{document}

\title{Almost (Weighted) Proportional Allocations for Indivisible Chores\thanks{The authors thank Haris Aziz, Herv\'e Moulin and Warut Suksompong for their valuable comments and suggested related works.}}


\author{Bo Li$^1$ \hspace{30pt} Yingkai Li$^2$ \hspace{30pt} Xiaowei Wu$^3$\\
$^1$Department of Computing, The Hong Kong Polytechnic University\\
\texttt{comp-bo.li@polyu.edu.hk}\\
$^2$Department of Computer Science, Northwestern University\\
\texttt{yingkai.li@u.northwestern.edu} \\
$^3$IOTSC, University of Macau\\
\texttt{xiaoweiwu@um.edu.mo}}

\maketitle

\begin{abstract}
In this paper, we study how to fairly allocate $m$ indivisible chores to $n$ (asymmetric) agents. 
	We consider (weighted) {\em  proportionality up to any item} (PROPX), and show that a (weighted) PROPX allocation always exists and can be computed efficiently. 
	For chores, we argue that PROPX might be a more reliable relaxation for proportionality by the facts that any PROPX allocation ensures 2-approximation of maximin share (MMS) fairness [Budish, 2011] for symmetric agents and of anyprice share (APS) fairness [Babaioff et al, 2021] for asymmetric agents. 
	APS allocations for chores have not been studied before the current work, and our result implies a 2-approximation algorithm. 
	Another by-product result is that an EFX and a weighted EF1 allocation for indivisible chores exist if all agents have the same ordinal preference, which might be of independent interest. 
	We then consider the partial information setting and design algorithms that only use agents' ordinal preferences to compute approximately PROPX allocations. 
	Our algorithm achieves 2-approximation for both symmetric and asymmetric agents, and the approximation ratio is optimal.
	Finally we study the price of fairness (PoF), i.e., the loss in social welfare by enforcing allocations to be (weighted) PROPX.
	We prove that the tight ratio for PoF is $\Theta(n)$ for symmetric agents and unbounded for asymmetric agents.
	

\end{abstract}

\maketitle

\input{intro.tex}

\input{preliminaries.tex}

\input{ALG1-top-cycle}

\input{Big-and-Take}

\input{Ordinal}

\input{PriceOfFairness}

\section{Conclusion}

In this paper, we study the fair allocation problem of indivisible chores under the fairness notion of proportionality up to any item (PROPX).
We show that PROPX allocations exist and can be computed efficiently for both symmetric and asymmetric agents.
The returned allocations achieve the optimal guarantee on the price of fairness.
As byproducts, our results imply a 2-approximate algorithm for APS allocations for general instances, and the existence of EFX and weighted EF1 allocations for IDO instances. 
Finally, we design the optimal approximation algorithms to compute (weighted) PROPX allocations with ordinal preferences. 

There are many future directions that are worth  effort. 
To name a few, as we have discussed, the existence or approximation of EFX and weighted EF1 allocations are less explored for chores than those for goods.
Furthermore, we proved that any WPROPX allocation is 2-approximate APS, but it does not have good guarantee for weighted MMS defined in \cite{conf/ijcai/0001C019}.
It is still unknown whether weighted MMS admits constant approximations.
Finally, it is also interesting to investigate the compatibility between PROPX and efficiency notions such as Pareto optimality. 



\bibliographystyle{abbrv}

\bibliography{propx}

\newpage

\appendix

\input{appendix}

\end{document}

%% file: intro.tex
\section{Introduction}

Fairness has drawn an increasing concern in broader areas including but not limited to mathematics, economics, and computer science. 
A fundamental problem is centered around allocating a set of resources (goods with non-negative utilities) or tasks (chores with non-positive utilities) to a number of agents in a way such that all agents are satisfied.
To capture the agents' preference for the allocation, people mostly study two solution concepts, {\em envy-freeness} \cite{foley1967resource} and {\em proportionality} \cite{steihaus1948problem}. 
Informally, an allocation is regarded as envy-free if nobody wants to exchange her items with any other agent in order to increase her utility.
Proportionality is weaker, 
which only requires that each agent has no smaller utility than her proportional share of all items. 
The traditional study of fair allocation mostly focused on divisible items (such as land and clean water),
where an envy-free or proportional allocation exists \cite{alon1987splitting,edward1999rental,conf/focs/AzizM16}.
However, the problem becomes trickier when the items are indivisible, due to the fact that an exact envy-free or proportional allocation is not guaranteed;
for example, consider the situation of allocating one item to two agents.

Accordingly, for indivisible items, an extensively studied subject is to investigate the extent to which the relaxations of these fairness definitions can be satisfied by either designing (approximation) algorithms or identifying hard instances to show the inherit difficulty of the problem so that no algorithm can be better than a certain performance.
For envy-freeness, two widely studied relaxations are {\em envy-free up to one item} (EF1) \cite{conf/sigecom/LiptonMMS04} and {\em envy-free up to any item} (EFX) \cite{journals/teco/CaragiannisKMPS19},
which were first proposed for goods.
EF1 allocations require that the envy disappears after removing one item from a bundle.
On the positive side, it is shown in \cite{conf/sigecom/LiptonMMS04} and \cite{DBLP:journals/corr/abs-2012-06788} that an EF1 allocation is guaranteed to exist and can be found efficiently.
On the negative side, some EF1 allocations can be relatively unfair even when the instance admits a fairer allocation.
EFX is proposed to improve fairness guarantee, where the envy is eliminated after removing any item from that bundle. 
Though EFX is fairer, it is still unknown whether such an allocation exists or not, except for some special cases.
Regarding proportionality, maximin fairness (MMS) \cite{conf/bqgt/Budish10} is arguably one of the most extensively studied relaxation.
It has been shown that for goods~\cite{journals/jacm/KurokawaPW18} and chores~\cite{conf/aaai/AzizRSW17}, an exact MMS allocation is not guaranteed to exist, but constant approximation algorithms are known.
Proportional up to one item (PROP1) allocation, which is weaker than EF1 and thus inherits many good properties of EF1 allocations, is studied in \cite{conf/sigecom/ConitzerF017,journals/orl/AzizMS20}.
Unfortunately, it is shown in \cite{journals/orl/AzizMS20} that a proportional up to any item  (PROPX) allocation is not guaranteed to exist for indivisible goods.

We follow this research trend and study PROPX allocations for indivisible chores with additive valuations. 
We attempt to provide a relatively extensive understanding for PROPX, including existence, computation, and price of fairness. 
Moreover, we are interested in two advanced settings, namely, {\em allocation with asymmetric agents} and {\em approximation with ordinal preferences}.

First, in many practical scenarios, agents do not have the same share (endowment for the case of goods and obligation for chores) to the system.
For example, people at leadership positions may be liable to undertake more responsibility than others in a company. 
This scenario is called {\em asymmetric} or {\em weighted} case.
Algorithms for computing weighted EF1 or approximately weighted MMS allocations have been proposed in goods~\cite{journals/jair/FarhadiGHLPSSY19,conf/ijcai/0001C019,conf/atal/ChakrabortyISZ20}.
Recently, AnyPrice Share (APS) fairness is introduced in \cite{babaioff2021competitive}, and the authors argued that APS is more suitable than MMS for asymmetric agents.
Before the current work, approximately APS fair algorithms for chores have not been studied. 

Second, motivated by practical applications where it is hard for the algorithm to collect complete information on agent preferences, using partial information to approximately compute fair allocations has attracted much attention recently.
Ordinal information is a typical setting, where the algorithm only knows each agent's ranking on the items without cardinal values.
Using ordinal preferences to compute approximately MMS allocations has been studied for goods~\cite{DBLP:conf/ijcai/AmanatidisBM16,DBLP:journals/corr/abs-2105-10064} and chores~\cite{DBLP:journals/corr/abs-2012-13884}.

%


\subsection{Main Results}

In a system of allocating $m$ indivisible chores to $n$ agents, each agent $i$ has an obligation share $s_i\geq 0$ and $\sum_{i}s_i = 1$.
When agents are symmetric $s_i = \frac{1}{n}$.
Informally, an allocation is called PROPX if for any agent, by removing an arbitrary item from her bundle, her cost is no more than her share in the system. 
We argue that for indivisible chores, PROPX may be a more reliable relaxation of (weighted) proportionality than MMS and APS.
This is because the existence of MMS/APS allocation is not guaranteed even with three symmetric agents. 
However, we show that (weighted) PROPX allocations always exist. 
Moreover, any (weighted) PROPX allocation ensures 2-approximation of MMS for symmetric agents and of APS for asymmetric agents; however, an MMS or APS allocation can be as bad as $\Theta(n)$-approximation regarding PROPX. 

To compute PROPX allocations, we provide two algorithms and each of them has its own merits.
The first algorithm is based on the widely studied {\em envy-cycle elimination} technique \cite{conf/sigecom/LiptonMMS04,DBLP:conf/aaai/BarmanBMN18}, 
which is recently adapted to chores and called {\em Top-trading} envy-cycle elimination \cite{DBLP:journals/corr/abs-2012-06788}. 
It is proved in \cite{DBLP:journals/corr/abs-2012-06788} that the algorithm always returns an EF1 allocation. 
We prove a stronger property --- when the agents have the same ordinal preference for the items (which is called IDO instances~\cite{DBLP:journals/corr/abs-1907-04505}, short for identical ordering), the returned allocation is EFX (and thus PROPX) if we carefully select the items to be allocated in each step. 
Then we show a general reduction such that if we have an algorithm for IDO instances, we can convert it to the general instances while guaranteeing PROPX. 
Similar reductions are widely used in the computation of approximate MMS allocations~\cite{DBLP:journals/aamas/BouveretL16,DBLP:conf/sigecom/BarmanM17,DBLP:journals/corr/abs-1907-04505}. 
This argument holds for the weighted setting even when the agents do not have the same share. 
Finally, note that the envy-cycle elimination algorithm also ensures a $4/3$-approximation of MMS.
We summarize our results in this part as follows, which is proved in Theorem \ref{thm:envy-cycle}.

\medskip
\noindent
{\bf Result 1.1} (Simultaneous Fairness){\bf .}
{\em For every general instance, there is an allocation that is simultaneously PROPX and $4/3$-approximate MMS.
In addition, when the instance is IDO, the allocation is also EFX.}

\medskip

The second algorithm also cares about the efficiency of an allocation. 
To improve the social welfare, intuitively, each item should be allocated to the agent with lowest cost on that item.
We incorporate this idea to the design of our algorithm, {\em bid-and-take}, and show that it actually guarantees the tight approximation ratio to the optimal social cost subject to the PROPX constraint, i.e., the price of fairness. 
Moreover, we show that this algorithm can be easily extended to compute a {\em weighted} PROPX allocation for the weighted instances when the agents may have unequal obligation shares. 
Let $n$ be the number of agents and $m$ be the number of items.
We have the following result, which is proved in Theorems \ref{thm:bid-and-take}, \ref{thm:aps}, \ref{theorem:unweighted-pof}, \ref{theorem:weighted-pof}.

\medskip
\noindent
{\bf Result 1.2} (PROPX and PoF){\bf .}
{\em The allocation returned by the bid-and-take algorithm is (weighted) PROPX and 2-approximate APS.
It also achieves the tight approximation ratio to the optimal social cost subject to the PROPX constraint. 
The tight bound for the price of fairness (PoF) regarding PROPX is $\Theta(n)$ for unweighted case, $\Theta(m)$ for weighted IDO case and unbounded for weighted case.}

\medskip

Last but not least, following recent works \cite{DBLP:conf/ijcai/AmanatidisBM16,DBLP:journals/corr/abs-2105-10064,DBLP:journals/corr/abs-2012-13884}, we consider the partial information setting where the algorithm only knows agents' ordinal preferences without exact cardinal values. 
The intuition behind our algorithms is as follows. 
We partition agents into two sets so that each agent in the first set gets a single but large item, 
then a standard algorithm, such as (weighted) round-robin, is called on the agents in the second half to evenly allocate the remaining small items.
Although the idea of splitting agents into two parts looks artificial, the approximation ratio turns out to be tight.
A by-product result is that a weighted EF1 allocation exists for IDO instances. 
Accordingly, we obtain the following result, proved in Lemma \ref{lemma:ordinal-hardness} and Theorem \ref{thm:ordinal:weighetd}.

\medskip
\noindent
{\bf Result 2} (Ordinal Information){\bf .}
{\em With ordinal preferences, our algorithm achieves 2-approximate weighted PROPX for asymmetric agents.
Moreover, approximation ratio is optimal: no algorithm is able to achieve a better-than-2 approximation using only ordinal information, even in the unweighted case.}

\subsection{Other Related Works}

The definition of EFX allocation for goods is first proposed in \cite{journals/teco/CaragiannisKMPS19}, after which a lot of effort has been devoted to proving its existence. 
Currently, we only know that an exact EFX allocation exists for some special cases \cite{journals/siamdm/PlautR20,conf/sigecom/ChaudhuryGM20,journals/corr/abs-2102-10654,conf/ijcai/AmanatidisBFHV20}, and the general existence is still unknown. 
Particularly, it is shown in \cite{journals/siamdm/PlautR20} that an EFX allocation exists for IDO instances with goods. 
Our work complements this result by showing an EFX allocation also exists for IDO instances for chores. 
While most literature study the special yet important case where agents have equal entitlement or obligation share to the items, there is also a fast growing recent literature on the more general model in which agents may have arbitrary, possibly unequal shares. 
For example, \cite{journals/jair/FarhadiGHLPSSY19} and \cite{conf/ijcai/0001C019} adapt MMS to this setting for goods and chores, respectively, and designe approximation algorithms accordingly. 
\cite{babaioff2021competitive} and \cite{journals/corr/abs-2103-04304} provide different generalizations of MMS to this case.
A weighted EF1 allocation is shown to exist for goods in \cite{conf/atal/ChakrabortyISZ20} and unknown for chores.
Recently, AnyPrice Share (APS) fairness is introduced in \cite{babaioff2021competitive}, and the authors provide $\frac{3}{5}$-approximation algorithms for the case of goods.
Our work complements this research trend
by studying weighted PROPX, and its implication on APS allocation for chores.

Besides fairness, social welfare, which is actually a competing criterion to fairness, is another important measure to evaluate allocations.
The loss in social welfare by enforcing allocations to be fair is quantitatively measured by the {\em price of fairness}.
Bounding the price of fairness for goods and chores are widely studied in the literature \cite{journals/ior/BertsimasFT11,conf/wine/CaragiannisKKK09,journals/tcs/HeydrichS15,conf/ijcai/BeiLMS19,DBLP:conf/wine/BarmanB020,hohne2021allocating}.
In this paper, we study the price of fairness for indivisible chores under (weighted) PROPX requirement, and show that our algorithm achieves the optimal ratio.


%% file: preliminaries.tex
\section{Model and Solution Concepts}
\label{sec:preliminaries}

We consider the problem of fairly allocating a set of $m$ indivisible chores $M$ to a group of $n$ agents $N$.
Each agent $i\in N$ has a cost function $c_i: 2^M \to \bR^+\cup \{0\}$.
The cost functions are assumed to be additive in the current work; that is, for any item set $S\subseteq M$, $c_i(S) = \sum_{j\in S}c_i(\{j\})$. When there is no confusion, we use $c_{ij}$ and $c_i(j)$ to denote $c_i(\{j\})$ for short. 
Without loss of generality, assume all the cost functions are normalized, i.e., $c_i(M) = 1$.
An allocation is represented by a partition of the items $\bX=(X_1,\cdots,X_n)$, where each each agent $i$ obtains $X_i$, and $X_i\cap X_j = \emptyset$ for all $i\neq j$ and $\cup_{i\in N}X_i = M$.
An allocation is called partial if $\cup_{i\in N}X_i \subseteq M$.
Let the social cost of the allocation $\bX$ be $\sc(\bX) = \sum_{i\in N} c_i(X_i)$.
We are particularly interested in a special setting, {\em identical ordering} (IDO). 
An instance is called identical ordering (IDO) if all agents agree on the same ranking of the items, i.e., $c_{i1}\geq \ldots \geq c_{im}$ for $i\in N$.
Note that in an IDO instance, the agents' cardinal cost functions can still be different. 




We next define {\em envy-freeness}, {\em proportionality} and their relaxations.
For ease of discussion, in this section, we focus on the case of symmetric agents, called {\em unweighted case}, and defer the definitions for asymmetric-agent case to Section~\ref{sec:bid-and-take}.

\begin{definition}[EF and PROP]
    An allocation $\bX$ is called envy-free (EF) if $c_i(X_i) \le c_i(X_j)$ for any $i,j\in N$.
    It is called proportional (PROP) if $c_i(X_i) \le \PROP_i$ for any $i\in N$, where $\PROP_i = (1/n) \cdot c_i(M)$ is agent $i$'s proportional share for all items.
\end{definition}

For normalized cost functions $\PROP_i = 1/n$, $\forall i\in N$.

    


\begin{definition}[EF1 and EFX]
    An allocation $\bX$ is called envy-free up to one item (EF1) if for any $i,j\in N$, there exists $e\in X_i$ such that $c_i(X_i \setminus \{e\}) \le  c_i(X_j)$.
    The allocation $\bX$ is called envy-free up to any item (EFX) if for any $i,j\in N$ and any $e\in X_i$, $c_i(X_i \setminus \{e\}) \le c_i(X_j)$.
\end{definition}

It is easy to see that any EFX allocation is EF1, but not vice versa. 
We adopt similar ideas to relax the definition of proportionality for allocating chores.

\begin{definition}[PROP1 and PROPX]
    For any $\alpha\geq 1$, an allocation $\bX$ is $\alpha$-approximate proportional up to one item ($\alpha$-PROP1) if for any $i\in N$, there exists $e\in X_i$ such that $c_i(X_i \setminus \{e\}) \le \alpha \cdot \PROP_i$.
    The allocation $\bX$ is $\alpha$-approximate proportional up to any item ($\alpha$-PROPX) if for any $i\in N$ and any $e\in X_i$, $c_i(X_i \setminus \{e\}) \le \alpha \cdot \PROP_i$.
    When $\alpha = 1$, allocation $\bX$ is PROP1 or PROPX, respectively.
\end{definition}

Similarly, any PROPX allocation is PROP1, but not vice versa. 
As we will see for any additive cost functions, 
PROPX allocations exist and can be found in polynomial time.
Thus we always focus on PROPX allocations, and all results directly apply for the weaker definition of PROP1.

Before presenting our results, we recall {\em maximin share fairness}, 
which is another widely adopted way to relax proportionality.

\begin{definition}[MMS]
    Let $\Pi(M)$ be the set of all $n$-partitions of $M$.
    For any agent $i \in N$, her maximin share (MMS) is defined as 
    \[
    \MMS_i = \min_{\bX\in \Pi(M)} \max_{j\in N} \{c_i(X_j)\}.
    \]
    For any $\alpha\geq 1$, allocation $\bX$ is called $\alpha$-approximate maximin share fair ($\alpha$-MMS) if $c_i(X_i) \le \alpha \cdot \MMS_i$ for all $i\in N$.
    When $\alpha = 1$, $\bX$ is MMS fair.
\end{definition}
It is not hard to verify the following inequality,
\begin{align}\label{eq:mms:bound}
    \MMS_i \ge \max\left\{\max_{j \in M}\{c_{ij}\}, \PROP_i\right\}, \forall i \in N.
\end{align}

\section{Warm-up: Unweighted Agents}
\label{sec;warmup}
\subsection{Properties of PROPX Allocations}
First, it is not hard to see that PROPX is strictly weaker than EFX.

\begin{lemma}\label{lem:efx implies propx}
Any EFX allocation is PROPX. 
\end{lemma}
\begin{proof}
For any EFX allocation $\bX$ and any agent $i$, $c_i(X_i \setminus \{e\}) \le c_i(X_j)$ for all $e\in X_i$ and $j \in N$. 
Summing up the inequalities for all $j$, we have $n \cdot c_i(X_i \setminus \{e\}) \le c_i(M)$. Thus $\bX$ is PROPX.
\end{proof}


Next, we argue, through the two lemmas below, that PROPX might be a more reliable relaxation for proportionality than MMS in practice. 

\begin{lemma}\label{lem:propx imply 2mms}
Any PROPX allocation is 2-MMS. 
\end{lemma}

\begin{proof}
We will prove a stronger argument here:
for any PROPX allocation $\bX$ and for any agent~$i$, either $|X_i| \le 1$ or $c(X_i) \le 2\cdot \PROP_i$ for any $e\in X_i$.
Then by Equation \eqref{eq:mms:bound}, \cref{lem:propx imply 2mms} holds. 
For any agent $i$, if $|X_i| \leq 1$, the claim holds trivially. 
If $|X_i| \ge 2$, letting $e_i = \min_{e\in X_i} \{c_i(e)\}$,
we have 
\begin{equation*}
	c_i(X_i\setminus \{e_i\}) \le \PROP_i,
\end{equation*}
and 
\[
c_i(e_i) \le c_i(X_i\setminus \{e_i\}) \le \PROP_i.
\]
Thus $c_i(X_i) = c_i(X_i\setminus \{e_i\}) +c_i(e_i) \le 2\cdot \PROP_i$.
\end{proof}

It is easy to verify that the approximation ratio in the lemma is tight.
Consider an instance with two identical agents and two identical items. 
Allocating both items to one of them is PROPX but only 2-MMS.



\begin{lemma}
\label{lem:mms-is-bad-propx}
There exists an MMS allocation that is $\Theta(n)$-PROPX.
\end{lemma}

\begin{proof}
Consider an instance with $n$ agents and $m=n$ items where $n$ is sufficiently large.
In this instance all agents have identical cost function for the items.
For each agent $i$, let $c_{i1} = n-1$ and $c_{ij} = 1$ for items $j = 2,\cdots, n$. 
Thus 
\[
\MMS_i = n-1 \text{\quad and \quad} \PROP_i = \frac{2(n-1)}{n}.
\]

Consider an allocation where $X_i = \{2, \cdots, n\}$ is allocated to some agent $i$. Note that $c_i(X_i) = n-1 = \MMS_i$ is MMS for agent~$i$.
However, this allocation is not fair regarding PPROX because $c_i(X_i \setminus \{e\}) = n-2 = \Theta(n) \cdot \PROP_i$ for any $e\in X_i$.
\end{proof}



Actually, no allocation can be worse than $n$-PROPX, as for any chore allocation instance, the most unfair allocation is to allocate all items to a single agent, which is $n$-PROPX.

%% file: ALG1-top-cycle.tex
\subsection{Existence and Computation}
\label{sec:Top-cycle}

In this subsection, we prove the existence of PROPX allocations for the IDO instances.
In \cref{sec:bid-and-take}, we will prove that this is without loss of generality via a reduction to extend all our results to general cost functions where agents may have different ordinal preferences. 

\paragraph{The Algorithm.}
The envy cycle elimination algorithm is first proposed in \cite{conf/sigecom/LiptonMMS04} for goods and is adapted to chores in \cite{DBLP:journals/corr/abs-2012-06788} recently.
Given any (partial) allocation $\bX=(X_1,\ldots,X_n)$, we say that agent $i$ {\em envies} $j$ if $c_i(X_i) > c_i(X_j)$ and \emph{most-envies} $j$ if $c_i(X_i) > c_i(X_j)$ and $j = \arg\min_{k\in N} c_i(X_k)$.
For any allocation $\bX$, we can construct a directed graph $G_X$, called {\em top-trading envy graph}, where the agents are nodes and there is a directed edge from $i$ to $j$ if and only if $i$ most-envies $j$.
A directed cycle $C=(i_1,\ldots, i_d)$ is referred as a {\em top-envy cycle}. 
For any top-envy cycle $C$, the {\em cycle-swapped allocation} $\bX^C$ is obtained by reallocating bundles backwards along the cycle. That is, $X^C_i = X_i$ if $i$ is not in $C$, and
\begin{align*}
    X^C_{i_j} = \begin{cases}
    X_{i_{j+1}} & \text{ for all $1 \le j \leq d-1$}\\
    X_{i_1} & \text{ for $j = d$}.
    \end{cases}
\end{align*}
The algorithm works by assigning, at each step, an unassigned item {\em with largest cost} to an agent who does not envy anyone else (i.e., a non-envious agent who is a ``sink'' node in the top-trading envy graph). 
If the top-trading envy graph $G_X$ does not have a sink, then it must have a cycle \cite{DBLP:journals/corr/abs-2012-06788}.
Then resolving the top-trading envy cycles guarantees the existence of a sink agent in the top-trading envy graph.
The full description of the algorithm is introduced in Algorithm \ref{alg:envy-cycle}.


\begin{algorithm}
	\caption{\textsf{Top-trading Envy Cycle Elimination Algorithm}\label{alg:envy-cycle}}
	\textbf{Input}: IDO instance with $c_{i1}\geq c_{i2}\geq \ldots \geq c_{im}$ for all $i\in N$.
	
	Initialize: $\bX=(X_1,\cdots,X_n)$ where $X_i \leftarrow \emptyset$ for all $i \in N$.
	
	\For{$j = 1,2,\ldots,m$}
	{	
		\If{there is no sink in $G_X$}
		{
		Let $C$ be any cycle in $G_X$.
		
		Reallocate the items according to $X^C$ (i.e., the cycle-swapped allocation). 
		}
		Choose a sink $k$ in the graph $G_X$ and update $X_k \gets X_k \cup \{j\}$.
	}
	
	\textbf{Output}: Allocation $\bX=(X_1,\ldots,X_n)$.
\end{algorithm}


Algorithm \ref{alg:envy-cycle} is the same with the one designed in \cite{DBLP:conf/aaai/BarmanBMN18,DBLP:journals/corr/abs-2012-06788} except that we allocate the item with largest cost to the sink agent at each step.
It is proved in \cite{DBLP:journals/corr/abs-2012-06788} that no matter which item is allocated, the returned allocation is EF1.
In the following, we show a stronger argument: if we select the item with largest cost to allocate at each step, we can guarantee that the allocation is EFX for IDO instances.



\begin{lemma}\label{lem:IDO:EFX}
	For any IDO instance, Algorithm \ref{alg:envy-cycle} returns an EFX allocation in polynomial time. 
\end{lemma}

We prove \cref{lem:IDO:EFX} in Appendix.
Combining \cref{lem:IDO:EFX,lem:efx implies propx}, we have that for IDO instances, there exists a PROPX allocation which can be computed in polynomial time by Algorithm \ref{alg:envy-cycle}.
It is proved in \cite{DBLP:conf/sigecom/BarmanM17} that the envy-cycle elimination algorithm, where the cycles are resolved arbitrarily, guarantees $4/3$-MMS for chores.
Since Algorithm \ref{alg:envy-cycle} only uses one particular way to resolve the cycles, it continues to ensure $4/3$-MMS.

\begin{lemma}\label{lem:4/3-MMS}{\em \cite{DBLP:conf/sigecom/BarmanM17}}
    Algorithm \ref{alg:envy-cycle} outputs a $4/3$-MMS allocation.
\end{lemma}

We complement this result by an example implying that the analysis of Lemma \ref{lem:4/3-MMS} is tight. 
Consider the following instance with $n$ agents and $2n+1$ items where all agents have the same cost function shown in Table \ref{tab:lem:4/3}, where $\epsilon = \frac{1}{6(n-1)}$.
It is not hard to verify that $\MMS_i = 1+2\epsilon$ for all agent $i$, where the corresponding partition $\bX^*$ has $X^*_j = \{j, 2n - j -1\}$ for all $j \leq n-1$, and $X^*_n = \{2n-1,2n,2n+1\}$.
Note that 
\[
c_i(X^*_j) = 1+2\epsilon
\quad\text{ and } \quad
c_i(X^*_n) = 1+\epsilon.
\]
However, the allocation returned by Algorithm \ref{alg:envy-cycle} in the first $2n$ steps is $X_i = \{i, 2n-i+1\}$ for all $i\in N$ and accordingly $c_i(X_i) = 1$.
In the last step, no matter which agent obtains item $2n+1$, her cost will be $4/3$.
Thus the allocation is $\frac{4}{3(1+2\epsilon)}$-MMS. 

\begin{table*}[h]
    \centering
    \begin{tabular}{c|ccccccccccc}
    		Item & 1 & $\ldots$ & $j$ & $\ldots$ & $n$ & $n+1$ & $\cdots$ & $n+j$ & $\cdots$ & $2n$ & $2n+1$ \\
    		\hline
    		Cost & $\frac{2}{3}$ & $\cdots$ & $\frac{2}{3} - (j-1)\epsilon$ & $\ldots$& $\frac{1}{2}$ & $\frac{1}{2}$ & $\ldots$ & $\frac{1}{2} -(j-1)\epsilon$ & $\cdots$ & $\frac{1}{3}$ & $\frac{1}{3}$
    	\end{tabular}
    \medskip
    \caption{A Tight Example for Lemma \ref{lem:4/3-MMS}}
    \label{tab:lem:4/3}
    
\end{table*}

We conclude this section by the following theorem. 
Note that the reduction to general instances will be provided in Lemma \ref{lem:reduction} in \cref{sec:bid-and-take}.

\begin{theorem}\label{thm:envy-cycle}
There exists an algorithm that given any general instance returns an allocation that is simultaneously PROPX and  $4/3$-MMS.
In addition, when the instance is IDO, the allocation is also EFX.
\end{theorem}

%% file: Big-and-Take.tex
\section{Computing Weighted PROPX Allocations}
\label{sec:bid-and-take}

In this section, we focus on the general case where the agents may have different shares for the items. 
In the weighted setting, each agent $i\in N$ has a share $s_i \geq 0$, and we have $\sum_{i\in N} s_i = 1$.
Intuitively, $s_i$ represents how much fraction of the chores should be completed by agent $i$.
Accordingly, the weighted proportionality of each agent is $\WPROP_i=s_i\cdot c_i(M)$.
Again, we assume all the cost functions are normalized, and thus $\WPROP_i=s_i$.

\begin{definition}[WPROP and WPROPX]
	For any $\alpha \geq 1$, an allocation $\bX$ is called $\alpha$-approximate weighted proportional ($\alpha$-WPROP) if $c_i(X_i) \leq \alpha \cdot \WPROP_i$ for all $i\in N$.
	An allocation $\bX$ is $\alpha$-approximate weighted proportional up to any item ($\alpha$-WPROPX) if $c_i(X_i \setminus \{e\}) \le \alpha\cdot \WPROP_i$ for any agent $i\in N$ and any item $e \in X_i$.
	When $\alpha = 1$, allocation $\bX$ is WPROP or WPROPX, respectively. 
\end{definition}

As we have mentioned in Section \ref{sec;warmup}, to design algorithms to compute PROPX and WPROPX allocations, it is without loss of generality to focus on the IDO instances. 
We present the lemma as follows and leave the formal proof to Appendix. 

\begin{lemma}\label{lem:reduction}
    If there exists a polynomial time algorithm that given any IDO instance computes an $\alpha$-WPROPX allocation,  
    then there exists a polynomial time algorithm that given any instance computes an $\alpha$-WPROPX allocation.
\end{lemma}

Actually, Lemma \ref{lem:reduction} is in the same spirit with the counterpart reductions in 
\cite{DBLP:journals/aamas/BouveretL16,DBLP:conf/sigecom/BarmanM17,DBLP:journals/corr/abs-1907-04505} which are designed for (unweighted) MMS.
Although their results do not directly work for ``up to one'' relaxations, using the same technique, we show how to extend it to PROPX allocations and weighted settings.

\subsection{Bid-and-Take Algorithm}

Note that the Top-trading Envy Cycle Elimination Algorithm is not able to compute a WPROPX allocation for weighted setting.
Instead, in this section, we present the \emph{bid-and-take} algorithm.
In our algorithm, the items are allocated from the highest to the lowest cost.
Moreover, each item is allocated to an \emph{active} agent that minimizes the current social cost, i.e., has minimum cost on the item among all active agents.
Initially all agents are active.
When the cumulative cost of an agent exceeds her proportional share, we inactivate her.

\begin{algorithm}
	\caption{\textsf{The Bid-and-Take Algorithm}\label{alg:bid-and-take}}
	\textbf{Input}: IDO instance with $c_{i1}\geq \ldots \geq c_{im}$ and $c_i(M) = 1$ for all $i\in N$, and shares of agents $0\leq s_1\leq s_2 \leq \ldots \leq s_n$ and $\sum_{i\in N} s_i = 1$.
	
	Let $X_i \leftarrow \emptyset$ for all $i \in N$ and $A \leftarrow N$ be the set of active agents.
	
	\For{$j = 1,2,\ldots,m$}
	{	
		Let $i = \arg\min_{i'\in A} \{ c_{i'j} \}$, and set $X_i \leftarrow X_i \cup \{ j \}$.
		
		\If{$c_i(X_i) > s_i$}
		{
		    $A \leftarrow A \setminus \{ i \}$.
		}
	}
	
	\textbf{Output}: Allocation $\bX=(X_1,\ldots,X_n)$.
\end{algorithm}

\begin{lemma}\label{theorem:wpropx-computation}
    Algorithm \ref{alg:bid-and-take} returns a WPROPX allocation in polynomial time for IDO instances.
\end{lemma}
\begin{proof}
    Note that in Algorithm~\ref{alg:bid-and-take}, an agent is turned into inactive as soon as $c_i(X_i) > s_i$, 
    and no additional item will be allocated to this agent.
    Hence to show that the final allocation $\bX = (X_1,\ldots,X_n)$ is WPROPX, it suffices to show that the algorithm allocates all items, 
    i.e., the set of active agents $A$ is non-empty when each item $j\in M$ is considered.
    Note that if all items are allocated, for any agent $i\in N$, we have $c_i(X_i \setminus \{e\}) \leq s_i$ for any item $e\in X_i$ by the fact that items are allocated in decreasing order of the cost.
	
	Next we show $A\neq \emptyset$ when considering each item $j\in M$.
	
	\begin{claim}\label{claim:active-inactive-agents}
	    At any point, for any active agents $i, i'\in A$, 
	    \[
	    c_i(X_{i'}) \geq c_{i'}(X_{i'}).
	    \]
	\end{claim}
	\begin{proof}
	    Since we allocate each item to the active agent that has smallest cost on the item, we have $c_{ij} \geq c_{i'j}$ for each item $j\in X_{i'}$ (because both agents $i$ and $i'$ are active when the item is allocated).
	    Hence we have $c_i(X_{i'}) \geq c_{i'}(X_{i'})$, and Claim \ref{claim:active-inactive-agents} holds.
	\end{proof}

	Suppose when we consider some item $j$, all agents are inactive, i.e., $A = \emptyset$.
	Let $i$ be the last agent that becomes inactive.
	At the moment when $i$ becomes inactive, we have
	\begin{equation*}
	c_i(M) > \sum_{i' \in N} c_i(X_{i'}) \geq
	\sum_{i' \in N} c_{i'}(X_{i'}) \geq 
	\sum_{i' \in N} s_{i'} = 1, 
	\end{equation*}
	where the first inequality follows as item $j$ has not been allocated yet, which is a contradiction with $c_i(M) = 1$.
\end{proof}

Combining the above lemma with Lemma~\ref{lem:reduction}, we have the following theorem.

\begin{theorem}\label{thm:bid-and-take}
    There exists an algorithm that computes a WPROPX allocation for any weighted instance in polynomial time.
\end{theorem}

\subsection{AnyPrice Share Fair Allocation}

Next we show that our result directly implies a 2-approximation algorithm for AnyPrice Share fair allocations. 
We first adapt the AnyPrice Share fairness defined in \cite{journals/corr/abs-2103-04304} for goods to chores.
The high-level idea is as follows. 
Each agent's weight is regarded as her loan to the system and she can obtain a reward by completing a chore to repay the loan. 
The agent's AnyPrice Share is then defined as the smallest cost she can guarantee by completing a set of chores that suffices to repay the loan when the items' rewards are adversarially set with a total reward of 1. 
Let $\cR=\{(r_1,\cdots,r_m) \mid r_j\ge 0 \text{ for all } j\in M \text{ and } \sum_{j\in M} r_j = 1\}$ be the set of item-reward vectors
that sum up to 1. 

\begin{definition}[AnyPrice Share]
    The AnyPrice Share (APS) of agent $i$ with weight $s_i$ is defined as
    \[
    \APS_i = \max_{(r_1,\cdots,r_m) \in \cR} \min_{S \subseteq M} \left\{c_i(S) \mid \sum_{j \in S} r_j \ge s_i \right\}.
    \]
    For any $\alpha \geq 1$, an allocation $\bX$ is $\alpha$-approximate AnyPrice Share fair ($\alpha$-APS) if $c_i(X_i) \leq \alpha \cdot \APS_i$ for any agent $i\in N$.
	When $\alpha = 1$, allocation $\bX$ is APS fair.
\end{definition}


Similar with Lemmas \ref{lem:propx imply 2mms} and \ref{lem:mms-is-bad-propx}, we have the following  Lemma \ref{lem:wpropx imply 2aps} regarding WPROPX and APS. 

\begin{lemma}\label{lem:wpropx imply 2aps}
Any WPROPX allocation is 2-APS; an APS allocation can be $\Theta(n)$-WPROPX. 
\end{lemma}

To prove Lemma \ref{lem:wpropx imply 2aps}, we use the following auxiliary lemma.

\begin{lemma}\label{lemma:aps-bound}
	For all agent $i\in N$, we have 
	\[
	\APS_i \ge \max\{s_i, \max_{j\in M}\{ c_{ij} \} \}.
	\]
\end{lemma}
\begin{proof}
	To show $\APS_i \ge s_i$, it suffices to find a reward vector $(r_1,\cdots,r_m)$ such that every set $S$ of items with reward no less than $s_i$ is at least of cost $s_i$.
	Thus, we can simply set $r_i = c_{ij}$.
	Similarly, to show $\APS_i \ge c_{ij^*}$, where $j^* = \arg \max_{j\in M}c_{ij}$, we set $r_{j^*} = 1$ and $r_{ij} = 0$ for $j\neq j^*$.
	Then the unique way to repay the loan is to complete chore $j^*$, which incurs cost $c_{ij^*}$.
\end{proof}

\begin{proofof}{Lemma \ref{lem:wpropx imply 2aps}}
For any WPROPX allocation $\bX$ and any agent $i$, we have $c_i(X_i\setminus\{j\}) \le s_i$ for any $j\in X_i$.
Thus $c_i(X_i) \le s_i + c_{ij}\le s_i + \max_{e\in M} \{c_{ie}\} \le 2\cdot \APS_i$,
where the last inequality follows from Lemma~\ref{lemma:aps-bound}.

To see the second claim, we recall the example in Lemma \ref{lem:mms-is-bad-propx}.
In that example there are $n$ agents and $m=n$ items, and for all agent $i\in N$ we have $c_{i1} = n-1$ and $c_{ij} = 1$ for items $j = 2,\cdots, n$.
Thus 
\[
\APS_i \ge \max_{j\in M} \{c_{ij}\} = n-1 \text{\quad and \quad} \PROP_i = \frac{2(n-1)}{n}.
\]
Thus allocating $\{2, \cdots, n\}$ to some agent $i$ is APS to her but has $\Theta(n)$ approximation regarding PROPX.
\end{proofof}

Combing Theorem \ref{thm:bid-and-take} and Lemma \ref{lem:wpropx imply 2aps}, we directly have the following result.
\begin{theorem}
\label{thm:aps}
    There is an algorithm that computes a 2-APS fair allocation for any weighted instance in polynomial time.
\end{theorem}

%% file: Ordinal.tex
\section{Ordinal Setting}
\label{sec:ordinal}

In this section, we investigate the extent to which we can approximately compute (weighted) PROPX allocations with ordinal preferences. 
Using ordinal formation to design fair allocation algorithms is studied in \cite{DBLP:conf/ijcai/AmanatidisBM16,DBLP:journals/corr/abs-2012-13884,conf/ijcai/Halpern21}.
Note that the problem becomes trivial if $m\le n$, because as long as every agent gets at most one item, the allocation is PROPX.
Thus in the following, assume $m>n$.
By Lemma \ref{lem:reduction}, it suffices to consider the IDO instances.

\subsection{Unweighted Setting}

To highlight the intuition, we first consider the unweighted case.
We present an algorithm (in Algorithm \ref{alg:ordinal:unweighted}) that always computes a $2$-PROPX allocation in polynomial time.
In the algorithm, we partition the agents into two groups: $N_1 = \{ 1,2,\ldots,\lfloor\frac{n}{2}\rfloor \}$ and $N_2 = N\setminus N_1$.
We first allocate each agent in $N_1$ a large item and then run a round-robin algorithm where each agent in $N_2$ takes turns to select an item with largest cost from the remaining items.


\begin{algorithm}
	\caption{\textsf{Ordinal Algorithm for Approximate PROPX Allocation} \label{alg:ordinal:unweighted}}
	\textbf{Input}: IDO instance with $c_{i1}\geq c_{i2}\geq \ldots \geq c_{im}$ for all $i\in N$.
	
	Initialize: $\bX=(X_1,\cdots,X_n)$ where $X_i \gets \emptyset$ for all $i \in N$.
	
	\For{$j = 1,2,\ldots,\lfloor\frac{n}{2}\rfloor$}
	{	
		$X_j \gets \{j\}$.
	}
	
	Let $i = 1$.
	
	\For{$j = \lfloor\frac{n}{2}\rfloor + 1,\lfloor\frac{n}{2}\rfloor + 2,\ldots,m$}
	{	
		$X_{i+\lfloor\frac{n}{2}\rfloor} \gets X_{i+\lfloor\frac{n}{2}\rfloor} \cup \{j\}$.\\
		$i \gets i \mod (\lfloor\frac{n}{2}\rfloor+1) + 1$.
	}
	
	\textbf{Output}: Allocation $\bX=(X_1,\ldots,X_n)$.
\end{algorithm}


\begin{theorem}\label{thm:ordinal:unweighetd}
Algorithm \ref{alg:ordinal:unweighted} returns a $2$-PROPX allocation in polynomial time.
\end{theorem}
\begin{proof}
It is straightforward that the algorithm runs in polynomial time.
Recall that $N_1 = \{ 1,2,\ldots,\lfloor\frac{n}{2}\rfloor \}$ and $N_2 = N\setminus N_1$.
Thus $|N_1|=|N_2|$ if $n$ is even and $|N_1|+1=|N_2|$ otherwise.
Let $(X_1,\cdots, X_n)$ be the returned allocation.
It is obvious that the allocation is PROPX for all $i \in N_1$ as $|X_i| = 1$.
Consider any agent $i\in N_2$.
Denote by $X_i = \{e_{1}, \cdots, e_{k}\}$ the items allocated to $i$, where
\begin{equation*}
	c_i(e_1)\geq c_i(e_2) \geq \ldots\geq c_i(e_{k}).
\end{equation*}
Since the items allocated to agents in $N_1$ are those with largest costs, we have
\begin{equation}
	c_i(e_1)\leq c_i(X_l) \text{\quad for all $l \in N_1$}. \label{eq:propx-ordinal-e1}
\end{equation}

Moreover, as the items $\{ \lfloor\frac{n}{2}\rfloor + 1,\ldots,m \}$ are allocated from the most costly to the least costly in a round-robin manner, we have
\[
c_i(X_i\setminus\{e_1\}) \le c_i(X_j) \text{\quad for all $j \in N_2\setminus \{i\}$}.
\]

Thus we have
\begin{align*}
    |N_2|\cdot c_i(X_i\setminus\{e_1\})
    &\le c_i(X_i\setminus\{e_1\}) + \sum_{j\in N_2\setminus \{i\}} c_i(X_j)\\
    &= \sum_{j\in N_2}c_i(X_j) - c_i(e_1)\\
    & = 1- \sum_{l\in N_1}c_i(X_l) - c_i(e_1) \\
    &\le 1- (|N_1|+1)\cdot c_i(e_1),
\end{align*}
where the last inequality follows from Inequality~\eqref{eq:propx-ordinal-e1}.
Thus
\begin{align*}
    c_i(X_i) & = c_i(e_1) + c_i(X_i\setminus \{e_1\})\\
    &\le c_i(e_1) + \frac{1}{|N_2|}\big(1- (|N_1|+1)\cdot c_i(e_1)\big)\\
    &\le \frac{1}{|N_2|} + \left(1- \frac{|N_1|+1}{|N_2|}\right)\cdot c_i(e_1)\\
    &\le \frac{1}{|N_2|}\leq \frac{2}{n} = 2\cdot \PROP_i,
\end{align*}
where the second last inequality follows from $|N_1|+1 \ge |N_2|$, and the last holds because $|N_2| \ge \frac{n}{2}$.
\end{proof}

Actually, regarding PROPX, the approximation ratio $2$ our algorithm achieves is the best possible for ordinal algorithms, proved in the following lemma.  

\begin{lemma}\label{lemma:ordinal-hardness}
	With only ordinal preferences, no algorithm can guarantee a better-than-2 approximation for PROPX.
\end{lemma}
\begin{proof}
	Consider an IDO instance with 2 agents and $m$ items, where $m$ is sufficiently large and $c_{i1}\ge \cdots \ge c_{im}$ for both $i\in \{1,2\}$.
	Without loss of generality, suppose item 1 is allocated to agent 1.
	If $|X_1| > 1$, consider the cardinal costs for agent 1, $c_{11} = 1$ and $c_{1j}=0$ for all $j>1$, and thus $c_1(X_1 \setminus\{e\}) = 2\cdot\PROP_1$ for any $e\in X_1\setminus\{1\}$.
	If $|X_1| = 1$, consider the cardinal costs for agent 2,  $c_{2j}=1/m$ for all $j$, and thus $c_2(X_2 \setminus\{e\}) = 1-2/m \approx 2\cdot\PROP_2$ for any $e\in X_2$.
\end{proof}

\subsection{Weighted Setting}
We next extend Algorithm \ref{alg:ordinal:unweighted} to the weighted setting.
Given an arbitrary weighted instance with $s_1\le \cdots \le s_n$, let
\begin{equation*}
	i^* = \max \{i \mid \sum_{j=1}^i s_j \le \frac{1}{2}\}.
\end{equation*}
We partition the agents into two groups: $N_1 = [i^*]$ and $N_2 = N \setminus N_1$.
Let $w_1 = \sum_{i\in N_1} s_i$.
It is not hard to see the following properties. 
\begin{itemize}
	\item By definition we have $w_1 \leq 1/2$.
	\item We have $i^* \geq n/2$ because $\sum_{i=1}^{n/2} s_i \leq 1/2$.
	\item We have $s_i \ge \frac{1}{2(i^*+1)}$ for all $i \in N_2$ because otherwise $s_{i^*+1} < \frac{1}{2(i^*+1)}$,
	which implies $\sum_{j=1}^{i^*+1}s_j \leq \frac{1}{2}$.
	In other words, $i^*+1$ should be included in $N_1$ as well, which is a contradiction.
\end{itemize}

As before, we assign each agent $j\in N_1$ a single item $j\in M$.
Recall that these are the $i^*$ items with the maximum costs.
Let $M_L = \{ 1,2,\ldots,i^* \}$ be these items, and $M_S = M\setminus M_L$ be the remaining items.
We call $M_L$ the \emph{large} items and $M_S$ the \emph{small} items.
Then we run a weighted version of round robin algorithm by repeatedly allocating an item to the agent $i\in N_2$ with the minimum $|X_i|/s_i$ until all items are allocated (see Algorithm~\ref{alg:ordinal:weighted}).
The weighted round robin algorithm is proved to ensure weighted EF1 for indivisible goods in \cite{conf/atal/ChakrabortyISZ20}.
In Lemma~\ref{lemma:wpropx-for-N2}, we prove that the weighted round robin algorithm also ensures weighted EF1 for IDO instance with chores. 

\smallskip

\begin{algorithm}
	\caption{\textsf{Ordinal Algorithm for Approximate WPROPX Allocation} \label{alg:ordinal:weighted}}
	\textbf{Input}: IDO instance, and the shares of agents $s_1\le \cdots \le s_n$.
	
	Initialize: $\bX=(X_1,\cdots,X_n)$ where $X_i \leftarrow \emptyset$ for all $i \in N$.
	
	Let $i^* = \max \{i \mid \sum_{j=1}^i s_j \le \frac{1}{2}\}$, $N_1 = \{1, \cdots, i^*\}$ and $N_2 = N \setminus N_1$.
	
	
	\For{$j = 1,2,\ldots,i^*$}
	{	
		$X_j \gets \{j\}$.
	}
	
	\For{$j = i^* + 1,\ldots,m$}
	{	
		Let $l = \arg\min_l \{ {|X_l|}/{s_l} \}$ where tie is broken by agent ID. \\
		$X_{l} \gets X_{l} \cup \{j\}$.
	}
	
	\textbf{Output}: Allocation $\bX=(X_1,\ldots,X_n)$.
\end{algorithm}

\smallskip

Our main result relies on the following technical lemma. 

\begin{lemma}\label{lemma:wpropx-for-N2}
     For every agent $j\neq i$, we have
    \[
    \frac{c_i( X_i \setminus \{i\} )}{s_i} \leq \frac{c_i(X_j)}{s_j}.
    \]
\end{lemma}
\begin{proof}
    Suppose $X_i = \{e_1,e_2,\ldots,e_k\}$, where $c_{i e_1} \geq c_{i e_2}\geq \ldots \geq c_{i e_k}$.
	Recall that $e_1 = i$.
	We define a real-valued function $\rho: (0,\frac{k}{s_i}] \rightarrow \bR^+$ as follows:
	\begin{equation*}
		\rho(\alpha) = c_{i e_t}, \text{\quad for $\frac{t-1}{s_i}< \alpha \leq \frac{t}{s_i}$ and $t\in [k]$.}
	\end{equation*}
	
	Thus, for all $t\in [k]$, we have $\frac{c_{i e_t}}{s_i} = \int_{\frac{t-1}{s_i}}^{\frac{t}{s_i}} \rho(\alpha) \, d\alpha$, and
		
	\begin{equation*}
		\frac{c_i(X_i\setminus \{i\})}{s_i} = \int_\frac{1}{s_i}^\frac{k}{s_i} \rho(\alpha) \, d\alpha.
	\end{equation*}

	Similarly, assume $X_j = \{e'_1,\ldots,e'_{k'}\}$, where $c_{i e'_1} \geq \ldots \geq c_{i e'_{k'}}$, and define
	\begin{equation*}
		\rho'(\alpha) = c_{i e'_t}, \text{\quad for } \frac{t-1}{s_j}< \alpha \leq \frac{t}{s_j}.
	\end{equation*}
	By definition we have
	\begin{equation*}
		\frac{c_i(X_j)}{s_j} = \int_0^\frac{k'}{s_j} \rho'(\alpha) \, d\alpha.
	\end{equation*}
	
	Recall that in Algorithm~\ref{alg:ordinal:weighted}, each item is allocated to the agent $i\in N_2$ with the minimum $|X_i|/s_i$. Thus we have
	\begin{equation*}
		\frac{k-1}{s_i} = \frac{|X_i| - 1}{s_i} \leq \frac{|X_j|}{s_j} = \frac{k'}{s_j},
	\end{equation*}
	where the inequality holds because otherwise item $e_k$ will not be allocated to agent $i$ in Algorithm~\ref{alg:ordinal:weighted}.
	
	Next we show that $\rho(\alpha) \leq \rho'(\alpha - \frac{1}{s_i})$.
	Consider the round when $\frac{|X_i|}{s_i}$ reaches $\alpha$.
	Suppose item $e_t$ is allocated to agent $i$ at this round, i.e., $\rho(\alpha) = c_{i e_t}$.
	Note that in this round, we must have $\frac{|X_j|}{s_j} \geq \alpha - \frac{1}{s_i}$.
	Because otherwise when item $e_t$ is considered we have $\frac{|X_j|}{s_j} < \alpha - \frac{1}{s_i} \leq \frac{|X_i|}{s_i}$, which means that item $e_t$ should not be allocated to agent $i$.
	In other words, the event ``$\frac{|X_j|}{s_j}$ reaches $\alpha - \frac{1}{s_i}$'' happens before the event ``$\frac{|X_i|}{s_i}$ reaches $\alpha$''.
	Since items are allocated from the most costly to the least costly, we have $\rho(\alpha) = c_{i e_t}\leq \rho'(\alpha - \frac{1}{s_i})$.
	
	Combining the above discussion, we have
	\begin{align*}
		\frac{c_i(X_i\setminus \{i\})}{s_i} & = \int_\frac{1}{s_i}^\frac{k}{s_i} \rho(\alpha) \, d\alpha
		\leq \int_\frac{1}{s_i}^\frac{k}{s_i} \rho'(\alpha - \frac{1}{s_i}) \, d\alpha \\
		& = \int_0^\frac{k-1}{s_i} \rho'(\alpha) \, d\alpha \leq \int_0^\frac{k'}{s_j} \rho'(\alpha) \, d\alpha = \frac{c_i(X_j)}{s_j},
	\end{align*}
	which proves the lemma. 
\end{proof}

Given the above lemma, 
we can obtain the following main result.
\begin{theorem}\label{thm:ordinal:weighetd}
Algorithm~\ref{alg:ordinal:weighted} computes a 2-WPROPX allocation in polynomial time for any given weighted instance.
\end{theorem}

\begin{proof}
	As before, it suffices to show that the allocation is 2-WPROP for agents in $N_2$, i.e., $c_i(X_i) \leq 2\cdot s_i$ for all $i\in N_2$.
	
	In Algorithm~\ref{alg:ordinal:weighted}, each agent $i\in N_2$ receives item $i\in M_S$ as her first item.
	Next we upper bound the total cost agent $i$ receives excluding item $i$.

	By Lemma~\ref{lemma:wpropx-for-N2}, we have 
	\begin{align*}
		\sum_{j\in N_2} \{ \frac{s_j}{s_i} \cdot c_i( X_i \setminus \{i\} ) \} 
		&\leq c_i( X_i \setminus \{i\} ) + \sum_{j\in N_2\setminus \{i\} } c_i(X_j) \\
		&= c_i(M_S) - c_{ii}.
	\end{align*}
	
	Reordering the above inequality, we have 
	\begin{align*}
	    c_i(X_i\setminus\{ i \}) &\leq \frac{s_i}{\sum_{j\in N_2} s_j}\cdot (c_i(M_S) - c_{ii}) \\
	    &\leq 2\cdot s_i\cdot (c_i(M_S) - c_{ii}),
	\end{align*}
    where the second inequality is because of $\sum_{j\in N_2} s_j \geq \frac{1}{2}$.

	Since every item in $M_L$ has cost at least $c_{ii}$ under the cost function of agent $i$, we have $c_i(M_L) \geq i^*\cdot c_{ii}$.
	Therefore we have
	\begin{align*}
		c_i(X_i) & = c_i( X_i\setminus \{i\} ) + c_{ii} \\
		& \leq  2\cdot s_i\cdot (c_i(M_S) - c_{ii}) + \frac{c_i(M_L)+c_{ii}}{i^* + 1} \\
		& \leq 2\cdot s_i \cdot (c_i(M_S) - c_{ii}) + 2\cdot s_i \cdot (c_i(M_L) + c_{ii})\\
		&\leq 2\cdot s_i,
	\end{align*} 
	where the second last inequality is due to $s_i \geq \frac{1}{2(i^*+1)}$.
\end{proof}

%% file: PriceOfFairness.tex
\section{Price of Fairness}

In this section, we show that the allocation returned by Algorithm~\ref{alg:bid-and-take} achieves the optimal price of fairness (PoF) among all WPROPX allocations.
Price of fairness is used to measure how much social welfare we lose if we want to maintain fairness among the agents.
Recall that we assume $c_i(M) = 1$ for all agents $i\in N$.
Let $\Omega(\cI)$ be the set of all WPROPX allocations for instance $\cI$, the price of fairness is defined as the worst-case ratio between the optimal (minimum) social cost $\opt(\cI)$ without any constraints and the social cost under WPROPX allocations:
\[
\POF = \max_{\cI} \min_{\bX \in \Omega(\cI)} \dfrac{\sc(\bX)}{\opt(\cI)}.
\]
Note that it is easy to see that for any instance $\cI$, $\opt(\cI)$ is obtained by allocating every item to the agent who has smallest cost on it. 
Moreover, the assumption of $c_i(M) = 1$ for all agent $i\in N$ is necessary: if there exist two agents $i$ and $j$ having very different values of $c_i(M)$ and $c_j(M)$, then we have unbounded \POF even in the unweighted and IDO setting because the social optimal allocation can allocate all items to one agent while WPROPX allocations cannot.

\begin{lemma}\label{lemma:sc-of-alg}
    Letting $\bX$ be the allocation returned by Algorithm~\ref{alg:bid-and-take}, we have $\sc(\bX) \leq 1$.
\end{lemma}

\begin{proof}
    Let $i\in N$ be the agent that receives the last item $m$.
    By Claim~\ref{claim:active-inactive-agents}, for all agent $j\neq i$, we have $c_i(X_j) \geq c_j(X_j)$, because agent $i$ is active throughout the whole allocation process.
    Hence we have $\sc(\bX) = \sum_{j\in N} c_j(X_j) \leq  \sum_{j\in N} c_i(X_j) = c_i(M) = 1$.
\end{proof}


We show that the allocation returned by Algorithm~\ref{alg:bid-and-take} achieves the optimal \POF.

\begin{theorem}\label{theorem:unweighted-pof}
    For the unweighted case ($s_i = 1/n$ for all $i\in N$), the PROPX allocation returned by Algorithm~\ref{alg:bid-and-take} achieves the optimal \POF, which is $\Theta(n)$.
\end{theorem}
\begin{proof}
    We first prove that the \POF is $\Omega(n)$ by giving an unweighted instance $\cI$ for which any PROPX allocation $\bX$ satisfies $\sc(\bX) \geq (n/6)\cdot \opt(\cI)$.
    In $\cI$, we have $n$ agents and $m = n$ items with cost functions shown in the table below. 
   	\begin{center}
    	\begin{tabular}{c|ccccc}
    		Agent & $c_{i1}$ & $\ldots$ & $c_{i, n-1}$ &  $c_{in}$ \\
    		\hline
    		$1$ & $2/n^2$ & $\ldots$ & $2/n^2$ & $1-{2(n-1)}/{n^2}$ \\
    		\hline
    		$2$ & $1/n$ & $\ldots$ & $1/n$ & $1/n$ \\
    		\hline
    		$\vdots$ & $\vdots$ & $\ddots$ & $\vdots$ & $\vdots$ \\
    		\hline
    		$n$ & $1/n$ & $\ldots$ & $1/n$ & $1/n$
    	\end{tabular}
    \end{center}

	For the above instance we have
	\begin{equation*}
	\opt(\cI) = (n-1)\cdot \frac{2}{n^2} + \frac{1}{n} < \frac{3}{n}.
	\end{equation*}
	
	However, any PROPX allocation $\bX$ allocates at most $n/2 + 1$ items to agent $1$ because otherwise the bundle she receives has cost larger than $1/n$ even after removing one item.
	Hence we have
    \begin{equation*}
    \sc(\bX) \geq (\frac{n}{2}+1)\cdot \frac{2}{n^2} + (\frac{n}{2} - 1) \cdot \frac{1}{n} > \frac{1}{2} \geq \frac{n}{6}\cdot \opt(\cI).
    \end{equation*}
    
    Next we show that for any unweighted instance $\cI$, the allocation $\bX$ computed by Algorithm~\ref{alg:bid-and-take} satisfies 
    \[
    \sc(\bX) \leq n\cdot \opt(\cI).
    \]
    
    By Lemma~\ref{lemma:sc-of-alg}, it suffices to consider the case when $\opt(\cI) < 1/n$.
    We show that in this case, we have $\sc(\bX) = \opt(\cI)$.
    This is because, before any agent becomes inactive, we always allocate an item to the agent that has smallest cost on the item, as in the social optimal allocation.
    Since $\opt(\cI) < 1/n$, Algorithm~\ref{alg:bid-and-take} never turns any agent into inactive, which implies that $\sc(\bX) = \opt(\cI)$.
\end{proof}

We note that the hard instance to show Theorem \ref{theorem:unweighted-pof} is IDO, which means the \POF is $\Theta(n)$ even for the unweighted IDO instances.
For weighted case, we have the following.

\begin{theorem}\label{theorem:weighted-pof}
    For the weighted case, we have unbounded \POF. For IDO instances, Algorithm~\ref{alg:bid-and-take} computes a WPROPX allocation with optimal \POF, which is $\Theta(m)$.
\end{theorem}
\begin{proof}
    We first show that the \POF is unbounded for weighted non-IDO instances by giving the following hard instance $\cI$ (with $s_1 = 1-\epsilon^2$ and $s_2 = \epsilon^2$) shown in the table below.
	It is easy to see that $\opt(\cI) = 2\epsilon$.
	However, since any WPROPX allocation $\bX$ allocates at most one item to agent $2$, we have $\sc(\bX) \geq 1/2$. 
	Since $\epsilon > 0$ can be arbitrarily close to $0$, we have an unbounded \POF for the weighted non-IDO instances.
    \begin{center}
    	\begin{tabular}{c|ccc}
    		Agent & $c_{i1}$ & $c_{i2}$ & $c_{i3}$ \\
    		\hline
    		$1$ & $0.5$ & $0.5$ & $0$ \\
    		\hline
    		$2$ & $\epsilon$ & $\epsilon$ & $1-2\epsilon$
    	\end{tabular}
    \end{center}
	
	Next we show that the \POF is $\Omega(m)$ for weighted IDO instances, by giving the following hard instance.
    \begin{center}
    	\begin{tabular}{c|ccccc}
    		Agent & $c_{i1}$ & $c_{i2}$ & $c_{i3}$ & $\ldots$ & $c_{im}$ \\
    		\hline
    		$1$ & $0.5$ & $0.5$ & $0$ & $\ldots$ & $0$ \\
    		\hline
    		$2$ & $1/m$ & $1/m$ & $1/m$ & $\ldots$ & $1/m$
    	\end{tabular}
    \end{center}
    
    It is easy to see that $\opt(\cI) = 2/m$.
	However, for $s_1 = 1-1/m^2$ and $s_2 = 1/m^2$, since any WPROPX allocation $\bX$ allocates at most one item to agent $2$, we have $\sc(\bX) \geq 1/2 \geq m/4\cdot \opt(\cI)$. 

    Finally, we show that the weighted IDO instances the allocation $\bX$ returned by Algorithm~\ref{alg:bid-and-take} satisfies $\sc(\bX) \leq m\cdot \opt(\cI)$.
    Recall that by Lemma~\ref{lemma:sc-of-alg}, we have $\sc(\bX) \leq 1$.
    On the other hand, for IDO instances,
    \begin{equation*}
        \opt(\cI) = \sum_{j\in M} \min_{i\in N} \{ c_{ij} \} \geq \min_{i\in N} \{ c_{i1} \} \geq \frac{1}{m} \geq \frac{1}{m}\cdot \sc(\bX).
    \end{equation*}
    Hence allocation $\bX$ achieves the asymptotically optimal \POF.
\end{proof}

%% file: appendix.tex
\section{Missing Proofs}

\subsection{Proof of Lemma \ref{lem:IDO:EFX}}

\begin{proof}
    Envy cycle elimination algorithm runs in polynomial time, which has been proved in \cite{conf/sigecom/LiptonMMS04,DBLP:journals/corr/abs-2012-06788}. In the following, we prove by induction that the returned allocation is EFX.
    First, if no item is allocated to any agent, the allocation is trivially EFX.
	Let~$\bX$ be a partial and EFX allocation at the beginning of any round in Algorithm \ref{alg:envy-cycle}. 
	Let $X_0$ be the set of all unassigned items. 
	We prove that at the end of this round, the new partial allocation is also EFX. We show this by proving the following two claims.
	
	\begin{claim}\label{claim:envy-cycle:add}
		Adding a new item to the allocation preserves EFX.
	\end{claim}
	
	Let $i$ be any sink agent in $G_X$. By definition $c_i(X_i) \le c_i(X_j)$ for all $j\in N$.
	Let $e$ be the item with largest cost in $X_0$, which will be added to $X_i$.
	Since the items are assigned from the most costly to least costly, and all agents have the same ordinal preference, $c_i(e') \ge c_i(e)$ for all $e' \in X_i$. Thus for any $e' \in X_i \cup \{e\}$ and any $j \neq i$,
	\[
	c_i(X_i \cup \{e\} \setminus \{e'\}) \le c_i(X_i \cup \{e\} \setminus \{e\}) = c_i(X_i) \le c_i(X_j).
	\]
	Thus Claim \ref{claim:envy-cycle:add} holds.
	
	\begin{claim}\label{claim:envy-cycle:resolve}
	Resolving a top-trading envy cycle preserves EFX.
	\end{claim}
	
	Suppose we reallocate the bundles according to a top-envy cycle $C=(i_1,\ldots, i_d)$ in $G_X$.
	For any agent $i$ who is not in the cycle, her bundle is not changed by the reallocation. 
	Although other bundles are reallocated, the items in each bundle is not changed and thus the cycle-swapped allocation is still EFX for agent $i$. 
	For any agent $i$ in $C$, she will obtain her best bundle in this partial allocation $\bX$, and hence the cycle-swapped allocation is EF for agent $i$. 
	Thus Claim \ref{claim:envy-cycle:resolve} holds.
	
	Combining the two claims, at the end of each round, the partial allocation remains EFX.
\end{proof}



\subsection{Proof of Lemma \ref{lem:reduction}}
\begin{proof}
    In the following, we explicitly write $\cI = (N, \mathbf{s}, M, \mathbf{c})$ to denote an instance with item set $M$, agent set $N$, weight vector $\mathbf{s} = (s_1,\cdots,s_n)$, and cost functions $\mathbf{c} = (c_1,\cdots,c_n)$.
    Given any instance $\cI = (N, \mathbf{s}, M, \mathbf{c})$, we construct an IDO instance $\cI' = (N, \mathbf{s}, M, \mathbf{c}')$ where $\mathbf{c}'=(c_1'\cdots,c_n')$ is defined as follows.
	Let $\sigma_i(j)\in M$ be the $j$-th most costly item under cost function $c_i$.
	Let
	\begin{equation*}
		c'_{ij} = c_{i \sigma_i(j)}.
	\end{equation*}
	
	Thus with cost functions $\mathbf{c}'$, the instance $\cI'$ is IDO, in which all agents $i$ has
	\begin{equation*}
		c'_{i1}\geq c'_{i2}\geq \ldots \geq c'_{im}.
	\end{equation*}
	
	Then we run the algorithm for IDO instances on instance $\cI'$, and get an $\alpha$-WPROPX allocation $\bX'$ for $\cI'$.
	By definition, for all agent $i\in N$ we have
	\begin{equation*}
	    c'_i(X'_i - e) \leq \alpha\cdot s_i, \qquad \forall e\in X'_i.
	\end{equation*}
	
	In the following, we use $\bX'$ to guide us on computing a $\alpha$-WPROPX allocation $\bX$ for instance $\cI$.
	
	Recall that in the IDO instance $\cI'$, for all agents, item $1$ has the maximum cost and item $m$ has the minimum cost.
	We initialize $X_i = \emptyset$ for all $i\in N$ and let $X_0 = M$ be the unallocated items.
	Sequentially for item $j$ from $m$ to $1$, we let the agent $i$ that receives item $j$ under allocation $\bX'$, i.e., $j\in X'_i$, pick her favourite unallocated item.
	Note that the order of items are well-defined in the IDO instance $\cI'$.
	Specifically, we move item $e = \arg\min_{e'\in X_0} \{ c_{i e'} \}$ from $X_0$ to $X_i$.
	Thus we have $|X_i| = |X'_i|$ for each agent $i\in N$.
	Furthermore, we show that there is a bijection $f_i: X_i \rightarrow X_i'$ such that for any item $e \in X_i$, we have $c_{ie} \le c'_{ij}$, where $j = f_i(e)$.
	Recall that $c'_{ij} \geq c'_{ik}$ for all $k \geq j$.
	By the way $c'_i$ is constructed, we know that there are at least $m - j + 1$ items that have cost at most $c'_{ij}$, under cost function $c_i$.
	Observe that when $e$ is chosen from $X_0$, we have $|X_0| \geq j$.
	Hence there must exists an item $e'$ in $X_0$ with cost $c_{i e'}\leq c'_{ij}$. Since $e$ has minimum cost among items in $X_0$ under cost function $c_i$, we have $c_{ie} \leq c'_{ij}$.
	Therefore, for any agent $i$ and any $e \in X_i$, we have
	\begin{equation*}
	    c_i(X_i - e) \le c'_i(X'_i - f_i(e)) \le \alpha\cdot s_i,
	\end{equation*}
	where the last inequality follows because $\bX'$ is $\alpha$-WPROPX for instance $\cI'$.
	Finally, it is easy to verify that the MMS benchmark in both instances are the same, while the cost of each agent in the general instance is smaller. 
	Thus the algorithm guarantees the same approximation ratio for MMS.
\end{proof}